\documentclass{llncs}

\usepackage{amssymb}                 

\usepackage{cite} 

\usepackage{tikz}
\usepackage{verbatim}
\usetikzlibrary{arrows,shapes}

\usepackage{amsmath}

\newtheorem{prop}{Proposition}
\newtheorem{lem}{Lemma}

\newcommand{\classNP}{{\sf NP}}

\newcommand{\ThreeSAT}{{\sc 3Sat}}

\begin{document}

\title{Matching Dynamics with Constraints\thanks{Supported by DFG Cluster of Excellence MMCI and grant Ho 3831/3-1.}}
\author{Martin Hoefer\inst{1} \and Lisa Wagner\inst{2}}
\institute{Max-Planck-Institut f\"ur Informatik and Saarland University, Germany\\ \email{mhoefer@mpi-inf.mpg.de} \and Dept. of Computer Science, RWTH Aachen University, Germany\\ \email{lwagner@cs.rwth-aachen.de}}
\date{}

\maketitle

\begin{abstract}
We study uncoordinated matching markets with additional local constraints that capture, e.g., restricted information, visibility, or externalities in markets. Each agent is a node in a fixed matching network and strives to be matched to another agent. Each agent has a complete preference list over all other agents it can be matched with. However, depending on the constraints and the current state of the game, not all possible partners are available for matching at all times. 

For correlated preferences, we propose and study a general class of hedonic coalition formation games that we call coalition formation games with constraints. This class includes and extends many recently studied variants of stable matching, such as locally stable matching, socially stable matching, or friendship matching. Perhaps surprisingly, we show that all these variants are encompassed in a class of ``consistent'' instances that always allow a polynomial improvement sequence to a stable state. In addition, we show that for consistent instances there always exists a polynomial sequence to every reachable state. Our characterization is tight in the sense that we provide exponential lower bounds when each of the requirements for consistency is violated.

We also analyze matching with uncorrelated preferences, where we obtain a larger variety of results. While socially stable matching always allows a polynomial sequence to a stable state, for other classes  different additional assumptions are sufficient to guarantee the same results. For the problem of reaching a \emph{given} stable state, we show \classNP-hardness in almost all considered classes of matching games.
\end{abstract}

\section{Introduction}

Matching problems are at the basis of many important assignment and allocation tasks in computer science, operations research, and economics. A classic approach in all these areas is \emph{stable matching}, as it captures distributed control and rationality of participants that arise in many assignment markets. In the standard two-sided variant, there is a set of men and a set of women. Each man (woman) has a preference list over all women (men) and strives to be matched to one woman (man). A (partial) matching $M$ has a blocking pair $(m,w)$ if both $m$ and $w$ prefer each other to their current partner in $M$ (if any). A matching $M$ is stable if it has no blocking pair. A large variety of allocation problems in markets can be analyzed using variants and extensions of stable matching, e.g., the assignment of jobs to workers
, organs to patients
, students to dormitory rooms, buyers to sellers, etc. In addition, stable matching problems arise in the study of distributed resource allocation problems in networks
.

In this paper, we study uncoordinated matching markets with dynamic matching constraints. An underlying assumption in the vast majority of works on stable matching is that matching possibilities are always available -- deviations of agents are only restricted by their preferences. In contrast, many assignment markets in reality are subject to additional (dynamic) constraints in terms of information, visibility, or externalities that prohibit the formation of certain matches (in certain states). Agents might have restricted information about the population and learn about other agents only dynamically through a matching process. For example, in scientific publishing we would not expect any person to be able to write a joint paper with a possible collaborator instantaneously. Instead, agents first have to get to know about each other to engage in a cooperation. Alternatively, agents might have full information but exhibit externalities that restrict the possibility to form certain matches. For example, an agent might be more reluctant to accept a proposal from the current partner of a close friend knowing that this would leave the friend unmatched.

Recent work has started to formalize some of these intuitions in generalized matching models with dynamic restrictions. For example, the lack of information motivates \emph{socially}~\cite{AskalidisIKMP13} or \emph{locally stable matching}~\cite{ArcauteV09}, externalities between agents have been addressed in \emph{friendship matching}~\cite{AnshelevichBH13}. On a formal level, these are matching models where the definition of blocking pair is restricted beyond the condition of mutual improvement and satisfies additional constraints depending on the current matching $M$ (expressing visibility/externalities/...). Consequently, the resulting stable states are supersets of stable matchings. Our main interest in this paper are convergence properties of dynamics that evolve from iterative resolution of such restricted blocking pairs. Can a stable state be reached from every initial state? Can we reach it in a polynomial number of steps? Will randomized dynamics converge (with probability 1 and/or in expected polynomial time)? Is it possible to obtain a particular stable state from an initial state (quickly)? These questions are prominent also in the economics literature (for a small survey see below) and provide valuable insights under which conditions stable matchings will evolve (quickly) in uncoordinated markets. Also, they highlight interesting structural and algorithmic aspects of matching markets.

Perhaps surprisingly, there is a unified approach to study these questions in all the above mentioned scenarios (and additional ones) via a novel class of coalition formation games with constraints. In these games, the coalitions available for deviation in a state are specified by the interplay of generation and domination rules. We provide a tight characterization of the rules that allow to show polynomial-time convergence results. They encompass all above mentioned matching models and additional ones proposed in this paper. In addition, we provide lower bounds in each model. 

\subsubsection*{Contribution and Outline}
A formal definition of stable matching games, socially stable, locally stable, and friendship matching can be found in Section~\ref{sec:prelim}. In addition, we describe a novel scenario that we term \emph{considerate matching}.

In Section~\ref{sec:correlated} we concentrate on stable matching with correlated preferences, in which each matched pair generates a single number that represents the utility of the match to both agents. Blocking pair dynamics in stable matching with correlated preferences give rise to a lexicographical potential function~\cite{AbrahamLMM08,AckermannGMRV11}. In Section~\ref{sec:constraints} we present a general approach on \emph{coalition formation games with constraints}. These games are hedonic coalition formation games, where deviating coalitions are characterized by sets of generation and domination rules. We concentrate on classes of rules that we term \emph{consistent}. For correlated preferences all matching scenarios introduced in Section~\ref{sec:prelim} can be formulated as coalition formation games with constraints and consistent rules. For games with consistent rules we show that from every initial coalition structure a stable state can be reached by a sequence of polynomially many iterative deviations. This shows that for every initial state there is always \emph{some} stable state that can be reached efficiently. In other words, there are polynomial ``paths to stability'' for all consistent games. Consistency relies on three structural assumptions, and we show that if either one of them is relaxed, the result breaks down and exponentially many deviations become necessary. This also implies that in consistent games random dynamics converge with probability 1 in the limit. While it is easy to observe convergence in expected polynomial time for socially stable matching, 
such a result is impossible for all consistent games due to exponential lower bounds for locally stable matching. The question for considerate and friendship matching remains an interesting open problem.

In Section~\ref{sec:NPC} we study the same question for a given initial state and a \emph{given stable state}. We first show that if there is a sequence leading to a given stable state, then there is also another sequence to that state with polynomial length. Hence, there is a polynomial-size certificate to decide if a given (stable) state can be reached from an initial state or not. Consequently, this problem is in \classNP\ for consistent games. We also provide a generic reduction in Section~\ref{sec:NPC} to show that it is \classNP-complete for all, socially stable, locally stable, considerate, and friendship matching, even with strict correlated preferences in the two-sided case. Our reduction also works for traditional two-sided stable matching with either correlated preferences and ties, or strict (non-correlated) preferences.

In Section~\ref{sec:general} we study general preferences with incomplete lists and ties that are not necessarily correlated. We show that for socially and classes of considerate and friendship matching we can construct for every initial state a polynomial sequence of deviations to a stable state. Known results for locally stable matching show that such a result cannot hold for all consistent games.



\subsubsection*{Related Work}
For a general introduction to stable matching and variants of the model we refer to textbooks in the area~\cite{Manlove13}. Over the last decade, there has been significant interest in dynamics, especially in economics, but usually there is no consideration of worst-case convergence times or computational complexity. While the literature is too broad to survey here, a few directly related works are as follows. If agents iteratively resolve blocking pairs in the two-sided stable marriage problem, dynamics can cycle~\cite{Knuth76}. On the other hand, there is always a ``path to stability'', i.e., a sequence of (polynomially many) resolutions converging to a stable matching~\cite{RothVV90}. If blocking pairs are chosen uniformly at random at each step, the worst-case convergence time is exponential. In the case of weighted or correlated matching, however, random dynamics converge in expected polynomial time~\cite{AckermannGMRV11,Mathieu10}. More recently, several works studied convergence time of random dynamics using combinatorial properties of preferences~\cite{HoffmanMP13}, or the probabilities of reaching certain stable matchings via random dynamics~\cite{BiroN13}.

In the roommates problem, where every pair of players is allowed to match, stable matchings can be absent, but deciding existence can be done in polynomial time~\cite{Irving85}. If there exists a stable matching, there are also paths to stability~\cite{DiamantoudiMX04}. Similar results hold for more general concepts like $P$-stable matchings that always exist~\cite{InarraLM08}. Ergodic sets of the underlying Markov chain have been studied~\cite{InarraLM10} and related to random dynamics~\cite{KlausFW10}. Alternatively, several works have studied the computation of (variants of) stable matchings using iterative entry dynamics~\cite{BlumRR97,BlumR02,BiroCF08}, or in scenarios with payments or profit sharing~\cite{BiroBGKP14,HoeferW13WINE,AnshelevichBH13}.

Locally stable matching was introduced by~\cite{ArcauteV09} in a two-sided job-market model, in which links exist only among one partition. More recently, we studied locally stable matching with correlated preferences in the roommates case~\cite{Hoefer13}, and with strict preferences in the two-sided case~\cite{HoeferW13}. For correlated preferences, we can always reach a locally stable matching using polynomially many resolutions of local blocking pairs. The expected convergence time of random dynamics, however, is exponential. For strict non-correlated preferences, no converging sequence might exist, and existence becomes \classNP-hard to decide. Even if they exist, the shortest sequence might require an exponential number of steps. These convergence properties improve drastically if agents have random memory.


Friendship and other-regarding preferences in stable matching games have been addressed by~\cite{AnshelevichBH13} in a model with pairwise externalities. They study existence of friendship matchings and bound prices of anarchy and stability in correlated games as well as games with unequal sharing of matching benefits. In friendship matching, agents strive to maximize a weighted linear combination of all agent benefits. In addition, we here propose and study considerate matching based on a friendship graph, where no agent accepts a deviation that deteriorates a friend. Such ordinal externalities have been considered before in the context of resource selection games~\cite{HoeferPPSV11}.

Our general model of coalition formation games with constraints is related to hedonic coalition formation games~\cite{Hajdukova06,Cechlarova08,BogomolnaiaJ02}. A prominent question in the literature is the existence and computational complexity of stable states (for details and references see, e.g., a recent survey~\cite{Woeginger13}).

\subsection{Preliminaries}
\label{sec:prelim}

A \emph{matching game} consists of a graph $G=(V,E)$ where $V$ is a set of vertices representing \emph{agents} and $E\subseteq \{\{u,v\}\mid u,v\in V, u\neq v\}$ defines the \emph{potential matching edges}. A \emph{state} is a matching $M\subseteq E$ such that for each $v\in V$ we have $|\{e~\mid~e\in M, v\in e\}|\leq 1$. An edge $e=\{u,v\}\in M$ provides \emph{utilities} $b_u(e), b_v(e) > 0$ for $u$ and $v$, respectively. If for every $e \in E$ we have some $b_u(e) = b_v(e)=b(e)>0$, we speak of \emph{correlated preferences}. If no explicit values are given, we will assume that each agent has an order $\succeq$ over its possible matching partners, and for every agent the utility of matching edges is given according to this ranking. In this case we speak of \emph{general preferences}. Note that in general, the ranking is allowed to be an incomplete list or to have ties. We define $B(M,u)$ to be $b_u(e)$ if $u\in e\in M$ and $0$ otherwise. A \emph{blocking pair} for matching $M$ is a pair of agents $\{u,v\} \not\in M$ such that each agent $u$ and $v$ is either unmatched or strictly prefers the other over its current partner (if any). A \emph{stable matching} $M$ is a matching without blocking pair.

Unless otherwise stated, we consider the \emph{roommates case} without assumptions on the topology of $G$. In contrast, the \emph{two-sided} or \emph{bipartite} case is often referred to as the \emph{stable marriage problem}. Here $V$ is divided into two disjoint sets $U$ and $W$ such that $E\subseteq \{\{u,w\}|~u\in U, w\in W\}$. Further we will consider matchings when  each agent can match with up to $k$ different agents at the same time.

In this paper, we consider broad classes of matching games, in which additional constraints restrict the set of available blocking pairs. The states that are resilient to such restricted sets of blocking pairs are a superset of stable matchings. Let us outline a number of examples that fall into this class and will be of special interest. \vspace{-0.3cm}

\subsubsection*{Socially Stable Matching} 
In addition to the graph $G$, there is a (social) \emph{network of links} $(V,L)$ which models static visibility. A state $M$ has a \emph{social blocking pair} $e=\{u,v\}\in E$ if $e$ is blocking pair and $e\in L$. Thus, in a social blocking pair both agents can strictly increase their utility by generating $e$ (and possibly dismissing some other edge thereby). A state $M$ that has no social blocking pair is a \emph{socially stable matching}. A \emph{social improvement step} is the resolution of such a social blocking pair, that is, the blocking pair is added to $M$ and all conflicting edges are removed.\vspace{-0.35cm}

\subsubsection*{Locally Stable Matching} 
In addition to $G$, there is a network $(V,L)$ that models dynamic visibility by taking the current matching into account. To describe stability, we assume the pair $\{u,v\}$ is \emph{accessible} in state $M$ if $u$ and $v$ have hop-distance at most 2 in the graph $(V,L\cup M)$, that is, the shortest path between $u$ and $v$ in $(V,L\cup M)$ is of length at most 2 (where we define the shortest path to be of length $\infty$, if $u$ and $v$ are not in the same connected component). A state $M$ has a \emph{local blocking pair} $e=\{u,v\}\in E$ if $e$ is blocking pair and $u$ and $v$ are accessible. Consequently, a \emph{locally stable matching} is a matching without local blocking pair. A \emph{local improvement step} is the resolution of such a local blocking pair, that is, the blocking pair is added to $M$ and all conflicting edges are removed. We also consider locally stable matchings where instead of the 2-hop-distance the $l$-hop-distance in $(V,L\cup M)$ defines the accessibility.\vspace{-0.35cm}

\subsubsection*{Considerate Matching} 
In this case, the (social) network $(V,L)$ symbolizes friendships and consideration. We assume the pair $\{u,v'\}$ is \emph{not accessible} in state $M$ if there is agent $v$ such that $\{u,v\} \in M$, and (a) $\{u,v\} \in L$ or (b) $\{v,v'\} \in L$. Otherwise, the pair is called accessible in $M$. Intuitively, this implies a form of consideration -- formation of $\{u,v'\}$ would leave a friend $v$ unmatched, so (a) $u$ will not propose to $v'$ or (b) $v'$ will not accept $u$'s proposal. A state $M$ has a \emph{considerate blocking pair} $e=\{u,v\}\in E$ if $e$ is blocking pair and it is accessible. A state $M$ that has no considerate blocking pair is a \emph{considerate stable matching}. A \emph{considerate improvement step} is the resolution of such a considerate blocking pair. \vspace{-0.35cm}

\subsubsection*{Friendship Matching} 
In this scenario, there are numerical values $\alpha_{u,v}\ge 0$ for every unordered pair $u,v\in V$, $u\neq v$, representing how much $u$ and $v$ care for each other's well-being. Thus, instead of the utility gain through its direct matching partner, $u$ now receives a \emph{perceived utility} $B_p(M,u) = B(M,u) + \sum_{v\in V\setminus\{u\}} \alpha_{u,v} B(M,v)$. In contrast to all other examples listed above, this definition requires cardinal matching utilities and cannot be applied directly on ordinal preferences. A state $M$ has a \emph{perceived blocking pair} $e=\{u,v\}\in E$ if $B_p(M,u) < B_p((M\setminus\{e'\mid e\cap e'\neq \emptyset\})\cup\{e\},u)$ and $B_p(M,v) < B_p((M\setminus\{e'\mid e\cap e'\neq \emptyset\})\cup\{e\},v)$. A state $M$ that has no perceived blocking pair is a \emph{perceived} or \emph{friendship stable matching}. A \emph{perceived improvement step} is the resolution of such a perceived blocking pair.


\section{Correlated Preferences}
\label{sec:correlated}
\subsection{Coalition Formation Games with Constraints}
\label{sec:constraints}

In this section, we consider correlated matching where agent preferences are correlated via edge benefits $b(e)$. In fact, we will prove our results for a straightforward generalization of correlated matching -- in correlated coalition formation games that involve coalitions of size larger than 2. In such a \emph{coalition formation game} there is a set $N$ of agents, and a set $\mathcal{C} \subseteq 2^N$ of hyper-edges, the \emph{possible coalitions}. We denote $n = |N|$ and $m = |\mathcal{C}|$. A \emph{state} is a \emph{coalition structure} $\mathcal{S}\subseteq \mathcal{C}$ such that for each $v\in N$ we have $|\{C\mid C\in \mathcal{S}, v\in C\}|\leq 1$. That is, each agent is involved in at most one coalition. Each coalition $C$ has a weight or benefit $w(C)>0$, which is the profit given to each agent $v \in C$. For a coalition structure $\mathcal{S}$, a \emph{blocking coalition} is a coalition $C\in \mathcal{C}\setminus \mathcal{S}$ with $w(C)>w(C_v)$ where $v \in C_v \in \mathcal{S}$ for every $v \in C$ which is part of a coalition in $\mathcal{S}$. Again, the resolution of such a blocking coalition is called an improvement step. A stable state or \emph{stable coalition structure} $\mathcal{S}$ does not have any blocking coalitions. Correlated matching games are a special case of coalition formation games where $\mathcal{C}$ is restricted to pairs of agents and thereby defines the edge set $E$. 

To embed the classes of matching games detailed above into a more general framework, we define \emph{coalition formation games with con\-straints}. For each state $\mathcal{S}$ we consider two sets of rules -- \emph{generation rules} that determine candidate coalitions, and \emph{domination rules} that forbid some of the candidate coalitions. The set of undominated candidate coalitions then forms the blocking coalitions for state $\mathcal{S}$. Using suitable generation and domination rules, this allows to describe socially, locally, considerate and friendship matching in this scenario.

More formally, there is a set $T \subseteq \{(\mathcal{T},C) \mid \mathcal{T} \subset \mathcal{C}, C \in \mathcal{C}\}$ of \emph{generation rules}. If in the current state $\mathcal{S}$ we have $\mathcal{T} \subseteq \mathcal{S}$ and $C \not\in \mathcal{S}$, then $C$ becomes a candidate coalition. For convenience, we exclude generation rules of the form $(\emptyset, C)$ from $T$ and capture these rules by a set $\mathcal{C}_g\subseteq \mathcal{C}$ of self-generating coalitions. A coalition $C \in \mathcal{C}_g$ is a candidate coalition for all states $\mathcal{S}$ with $C\not\in\mathcal{S}$. In addition, there is a set $D\subseteq\{(\mathcal{T},C)\mid \mathcal{T}\subset \mathcal{C},C\in \mathcal{C}\}$ of \emph{domination rules}. If $\mathcal{T} \subseteq \mathcal{S}$ for the current state $\mathcal{S}$, then $C$ must not be inserted. To capture the underlying preferences of the agents, we assume that $D$ always includes at least the set $D_w=\{(\{C_1\},C_2)\mid  w(C_1) \ge w(C_2),C_1 \cap C_2 \neq \emptyset, C_1\neq C_2\}$ of all weight domination rules. 

The undominated candidate coalitions represent the blocking coalitions for $\mathcal{S}$. In particular, the latter assumption on $D$ implies that a blocking coalition must at least yield strictly better profit for every involved agent. Note that in an improvement step, one of these coalitions is inserted, and every coalition that is dominated in the resulting state is removed. By assumption on $D$, we remove at least every overlapping coalition with smaller weight. A coalition structure is stable if the set of blocking coalitions is empty. 

Note that we could also define coalition formation games with constraints for general preferences. Then $D_w=\{(\{C_1\},C_2)\mid C_1 \cap C_2 \neq \emptyset, C_1\neq C_2, \exists v\in C_1: w_v(C_1) \ge w_v(C_2)\}$. However, a crucial point in our proofs is that in a chain of succeeding deletions no coalition can appear twice. This is guaranteed for correlated preferences as coalitions can only be deleted by more worthy ones. For general preferences on the other hand there is no such argument.

In the following we define \emph{consistency} for generation and for domination rules. This encompasses many classes of matching cases described above and is key for reaching stable states (quickly).

\begin{definition}
The generation rules of a coalition formation game with constraints are called \emph{consistent} if $T \subseteq \{(\{C_1\},C_2)\mid~ C_1\cap C_2\neq \emptyset\}$, that is, all generation rules have only a single coalition in their precondition and the candidate coalition shares at least one agent.
\end{definition}

\begin{definition}
The domination rules of a coalition formation game with constraints are called \emph{consistent} if $D\subseteq\{(\mathcal{S},C)\mid \mathcal{S}\subset \mathcal{C},C\in \mathcal{C}, C\notin \mathcal{S}, \exists S\in \mathcal{S}: S\cap C\neq \emptyset \}$, that is, at least one of the coalitions in $\mathcal{S}$ overlaps with the dominated coalition. Note that weight domination rules are always consistent.
\end{definition}



\begin{theorem}\label{polyTime}
In every correlated coalition formation game with constraints and consistent generation and domination rules, for every initial structure $\mathcal{S}$ there is a sequence of polynomially many improvement steps that results in a stable coalition structure. The sequence can be computed in polynomial time.
\end{theorem}
\begin{proof}
At first we analyze the consequences of consistency in generation and domination rules. For generation rules we demand that there is only a single precondition coalition and that this coalition overlaps with the candidate coalition. Thus if we apply such a generation rule we essentially replace the precondition with the candidate. The agents in the intersection of the two coalitions would not approve such a resolution if they would not improve. Therefore, the only applicable generation rules are those where the precondition is of smaller value than the candidate.

Now for domination rules we allow an arbitrary number of coalitions in the precondition, but at least one of them has to overlap with the dominated coalition. In consequence a larger set of coalitions might dominate a non-existing coalition, but to remove a coalition they can only use the rules in $D_w$. That is due to the fact that when a coalition $C$ already exists, the overlapping coalition of the precondition cannot exist at the same time. But this coalition can only be created if $C$ does not dominate it. Especially $C$ has to be less worthy than the precondition. Thus the overlapping precondition alone can dominate $C$ via weight.

The proof is inspired by the idea of the edge movement graph~\cite{Hoefer11Proc}. Given a coalition formation game with consistent constraints and some initial coalition structure $\mathcal{S}_0$, we define an object movement hypergraph 
\[G_{mov}=(V,V_g,T_{mov},D_{mov}).\]
A coalition structure corresponds to a marking of the vertices in $G_{mov}$. The vertex set is $V=\{v_C\mid C\in\mathcal{C}\}$, and $V_g=\{v_C\mid C\in\mathcal{C}_g\}$ the set of vertices which can generate a marking by themselves. The directed exchange edges are $T_{mov}=\{(v_{C_1},v_{C_2})\mid (\{C_1\},C_2)\in E, w(C_1)<w(C_2)\}$. The directed domination hyperedges are given by $D_{mov}=D_1\cup D_w$, where $D_1= \{(\{v_S\mid S\in \mathcal{S}\},v_C)\mid (\mathcal{S},C)\in D\}$. This covers the rule that a newly inserted coalition must represent a strict improvement for all involved agents. The initial structure is represented by a marking of the vertices $V_0=\{v_C\mid C\in \mathcal{S}_0\}$.

We represent improvement steps by adding, deleting, and moving markings over exchange edges to undominated vertices of the object movement graph. Suppose we are given a state $\mathcal{S}$ and assume we have a marking at every $v_C$ with $C \in \mathcal{S}$. We call a vertex $v$ in $G_{mov}$ currently \emph{undominated} if for every hyperedge $(U,v) \in D_{mov}$ at least one vertex in $U$ is currently unmarked. An improvement step that inserts coalition $C$ is represented by marking $v_C$. For this $v_C$ must be unmarked and undominated. We can create a new marking if $v_C \in V_g$. Otherwise, we must move a marking along an exchange edge to $v_C$. Note that this maps the generation rules correctly as we have seen, that we exchange the precondition for the candidate. To implement the resulting deletion of conflicting coalitions from the current state, we delete markings at all vertices which are now dominated through a rule in $D_{mov}$. That is, we delete markings at all vertices $v$ with $(U,v)\in D$ and every vertex in $U$ marked. 

Observe that $T_{mov}$ forms a DAG as the generation of the candidate coalition deletes its overlapping precondition coalition and thus the rule will only be applied if the candidate coalition yields strictly more profit for every agent in the coalition.
	
\begin{lem}
	\label{lem:objectMove}
	The transformation of markings in the object movement graph correctly mirrors the improvement dynamics in the coalition formation game with constraints.
\end{lem}
\begin{proof}
Let $\mathcal{S}$ be a state of the coalition formation game and let $C$ be a blocking coalition for $\mathcal{S}$. Then $C$ can be generated either by itself (that is, $C\in \mathcal{C}_g$) or through some generation rule with fulfilled precondition $C' \in \mathcal{S}$, and is not dominated by any subset of $\mathcal{S}$ via $D$. Hence, for the set of marked vertices $V_{\mathcal{S}}=\{v_C \mid C \in \mathcal{S}\}$ it holds that $v_C$ can be generated either because $v_C\in V_g$ or because there is a marking on some $v_{C'}$ with $(v_{C'},v_C)\in T_{mov}$, and is further not dominated via $D$. Hence, we can generate a marking on $v_C$. It is straightforward to verify that if $v_C$ gets marked, then in the resulting deletion step only domination rules of the form $\{(\{v_S\},v_T)\mid S,T\in\mathcal{C}, S\cap T\neq \emptyset$ and $w(S)\geq w(T)\}$ are relevant. Thus, deletion of markings is based only on overlap with the newly inserted coalition $C$. These are exactly the coalitions we lose when inserting $C$ in $\mathcal{S}$.

Conversely, let $V_{\mathcal{S}}$ be a set of marked vertices of $G_{mov}$ such that $\mathcal{S}=\{C\mid v_C\in V_{\mathcal{S}}\}$ does not violate any domination rule (i.e., for every $(\mathcal{U},C) \in D$, we have $\mathcal{U} \not\subseteq \mathcal{S}$ or $C \not\in \mathcal{S}$). Then $\mathcal{S}$ is a feasible coalition structure. Now if $v_C$ is an unmarked vertex in $G_{mov}$, then $C \notin \mathcal{S}$. Furthermore, assume $v_C$ is undominated and can be marked, because $v_C\in V_g$ or because some marking can be moved to $v_C$ via an edge in $T_{mov}$. Thus for every $\{\mathcal{S}_C,C\}\in D$ $v_C$ being undominated implies $\mathcal{S}_C\not\subset\mathcal{S}$. The property that $v_C$ can be marked implies that $C$ is self-generating or can be formed from $\mathcal{S}$ using a generation rule. Hence $C$ is a blocking coalition in $\mathcal{S}$. The insertion $C$ again causes the deletion of exactly the coalitions whose markings get deleted when $v_C$ is marked.\qed
\end{proof}
%

To show the existence of a short sequence of improvement steps we consider two phases.
\begin{description}
 \item [\bf Phase 1] In each round we check whether there is an exchange edge from a marked vertex to an undominated one. If this is the case, we move the marking along the exchange edge and start the next round. Otherwise for each unmarked, undominated $v\in V_g$ we compute the set of reachable positions. This can be done by doing a forward search along the exchange edges that lead to an unmarked undominated vertex. Note that the vertex has to remain undominated when there are the existing markings and a marking on the source of the exchange edge. If we find a reachable position that dominates an existing marking, we create a marking at the associated $v\in V_g$ and move it along the exchange edges to the dominating position. Then we start the next round. If we cannot find a reachable position which dominates an existing marking, we switch to Phase 2.
 \item [\bf Phase 2] Again we compute all reachable positions from $v\in V_g$.  We iteratively find a reachable vertex $v_C$ with highest weight $w(C)$, generate a marking at the corresponding $v \in V_g$ and move it along the path of exchange edges to $v_C$. We repeat this phase until no reachable vertex remains.
\end{description}

To prove termination and bound the length, we consider each phase separately. In Phase 1 in each round we replace an existing marking by a marking of higher value either by using an exchange edge or by deleting it through domination by weight. Further the remaining markings either stay untouched or get deleted. Now the number of improvements that can be made per marking is limited by $m$ and the number of markings is limited by $n$. Hence, there can be at most $mn$ rounds in Phase 1. Additionally, the number of steps we need per round is limited $m$ again, as we move the marking along the DAG structure of exchange edges. Thus, phase 1 generates a total of $O(n\cdot m^2)$ steps.

If in Phase 1 we cannot come up with an improvement, there is no way to (re)move the existing markings, no matter which other steps are made in subsequent iterations. This relies on the fact that the presence of additional markings can only restrict the subgraph of reachable positions. For the same reason, iteratively generating the reachable marking of highest weight results in markings that cannot be deleted in subsequent steps. Thus, at the end of every iteration in Phase 3, the number of markings is increased by one, and all markings are un(re)movable. Consequently, in Phase 2 there are $O(m\cdot n)$ steps.

For computation of the sequence, the relevant tasks are constructing the graph $G_{mov}$, checking edges in $T_{mov}$ for possible improvement of markings, or constructing subgraphs and checking connectivity of single vertices to $V_g$. Obviously, all these tasks can be executed in time polynomial in $n$, $m$, $|T|$ and $|D|$ using standard algorithmic techniques.
\qed
\end{proof}

Next, we want to analyze whether consistency of generation and domination rules is necessary for the existence of short sequences or can be further relaxed.

\begin{prop}\label{l>2}
If the generation rules contain more than one coalition in the precondition-set, there are instances and initial states such that every sequence to a stable state requires an exponential number of improvement steps.
\end{prop}
The proof uses a coalition formation game with constraints and inconsistent generation rules obtained from locally stable matching, when agents are allowed to match with partners at a hop distance of at most $\ell = 3$ in $(V, L \cup M)$. For this setting in~\cite[Theorem 3]{Hoefer13} we have given an instance such that every sequence to a stable state requires an exponential number of improvement steps. Note that the example is minimal in the sense that now we have at most 2 coalitions in the precondition-set. The detailed proof can be found in the appendix.

\begin{prop}\label{generationInconsistent}
If the generation rules have non-overlapping precondition- and target-coalitions, there are instances and initial states such that every sequence to a stable state requires an exponential number of improvement steps.
\end{prop}
The construction used for the proof exploits the fact that if precondition- and target-coalition do not overlap the precondition can remain when the target-coalition is formed. Then the dynamics require additional steps to clean up the leftover precondition-coalitions which results in an exponential blow-up. The entire proof as well as a sketch of the resulting movement graph can be found in the appendix.

\begin{prop}\label{cycle}
If the domination rules include target-coalitions that do not overlap with any coalition in the precondition, there are instances and starting states such that every sequence cycles. 
\end{prop}
\begin{proof}
 Consider the following small example with $N=\{1,\ldots,6\}$, $\mathcal{C}$ $=\{\{1,2\},$ $\{3,4\},$ $\{5,6\}\}$, $\mathcal{C}_g = \mathcal{C}$, and weights $w(C) = 1$ for all $C \in \mathcal{C}$. There are no generation rules: $T=\emptyset$ (in addition to $\mathcal{C}_g = \mathcal{C}$). For the domination rules, we consider non-overlapping coalitions in precondition and target: 
\[ D = \{(\{\{1,2\}\},\{3,4\}),(\{\{3,4\}\},\{5,6\}),(\{\{5,6\}\},\{1,2\})\}\enspace.\]
The initial state is $\mathcal{C}_{start}=\{1,2\}$.

Now with $\{1,2\}$ existing, $\{3,4\}$ is dominated and cannot be formed although it is a candidate coalition. The other candidate coalition $\{5,6\}$ is undominated and represents the unique improvement step. As $\{5,6\}$ dominates $\{1,2\}$ (but not vice versa), we lose $\{1,2\}$ when $\{5,6\}$ is formed. Now $\{4,3\}$ is the unique undominated candidate coalition and is formed. Thereby, we lose $\{5,6\}$, and $\{1,2\}$ becomes undominated. Now $\{1,2\}$ is formed, $\{4,3\}$ is deleted, and the cycle is complete.
\qed
\end{proof}


Consistent generation and domination rules arise in a large variety of settings, not only in basic matching games but also in some interesting extensions.

\begin{corollary}\label{embedding} Consistent generation and domination rules are present in 
\begin{itemize}
\item locally stable matching if agents can create $k=1$ matching edges and have lookahead $\ell = 2$ in $G=(V,M\cup L)$.
\item socially stable matching, even if agents can create $k \ge 1$ matching edges.
\item considerate matching, even if agents can create $k \ge 1$ matching edges.
\item friendship matching, even if agents can create $k \ge 1$ matching edges.
\end{itemize}
\end{corollary}

Due to space restrictions we cannot give a detailed description of the embedding into coalition formation games with constraints. In most cases the embedding is quite straightforward. Agents and edge set are kept as well as the benefits. The generation and domination rules often follow directly from the definitions. However, we want to shortly discuss the more complex mapping of $k$-matching for $k>1$ into coalition formation games with constraints. By definition no agent is allowed to participate in more than one coalition at the same time. Thus we cannot directly embed $k$-matching. Instead we have to map every agent $u$ to $k$ copies $u_1,\ldots,u_k$ who can match one vertex each. With these ``independent'' copies we now encounter the problem that $\{u_i,v_j\}$ and $\{u_{i'},v_{j'}\}$ should not exist simultaneously. This issue can easily be handled via the domination rules, but, as $\{u_i,v_j\}$ and $\{u_{i'},v_{j'}\}$ do not share any agents, rules of the form $(\{\{u_i,v_j\}\},\{u_{i'},v_{j'}\})$ would not be consistent. Thus for every edge $\{u,v\}$ we introduce an auxiliary vertex $a_{u,v}$. Potential coalitions then are given by $\{u_i,v_j,a_{u,v}\}$ instead of $\{u_i,v_j\}$. The exact embedding for every type of game can be found in the appendix. Additionally an exemplar proof for correctness is stated.

Unlike for the other cases, for locally stable matching we cannot guarantee consistent generation rules if we increase the number of matching edges. The same holds for lookahead $>2$. In both cases the accessibility of an edge might depend on more than one matching edge. There are exponential lower bounds in~\cite{Hoefer13,HoeferW13} for those extensions which proves that it is impossible to find an embedding with consistent rules even with the help of auxiliary constructions. 

\subsection{Reaching a Given Matching}
\label{sec:NPC}

In this section we consider the problem of deciding reachability of a \emph{given} stable matching from a given initial state. We first show that for correlated coalition formation games with constraints and consistent rules, this problem is in \classNP. If we can reach it and there exists a sequence, then there always exists a polynomial-size certificate due to the following result.

\begin{theorem}\label{alwaysShortSequence}
In a correlated coalition formation game with constraints and consistent generation and domination rules, for every coalition structure $\mathcal{S}^*$ that is reachable from an initial state $\mathcal{S}_0$ through a sequence of improvement steps, there is also a sequence of polynomially many improvement steps from $\mathcal{S}_0$ to $\mathcal{S}^*$.
\end{theorem}

For the proof we analyze an arbitrary sequence of improvement steps from $\mathcal{S}_0$ to $\mathcal{S}^*$ and show that, if the sequence is too long, there are unnecessary steps, that is, coalitions are created and deleted without making a difference for the final outcome. By identifying and removing those superfluous steps we can reduce every sequence to one of polynomial length. The detailed proof can be found in the appendix. 

For locally stable matching, the problem of reaching a given locally stable matching from a given initial matching is known to be \classNP-complete~\cite{HoeferW13}. Here we provide a generic reduction that shows \classNP-completeness for socially, locally, considerate, and friendship matching, even in the two-sided case. Surprisingly, it also applies to ordinary two-sided stable matching games that have either correlated preferences with ties, or non-correlated strict preferences. Observe that the problem is trivially solvable for ordinary stable matching and correlated preferences without ties, as in this case there is a unique stable matching that can always be reached using the greedy construction algorithm~\cite{AckermannGMRV11}.

\begin{theorem}\label{np}
It is \classNP-complete to decide if for a given matching game, initial matching $M_0$ and stable matching $M^*$, there is a sequence of improvement steps leading form $M_0$ to $M^*$. This holds even for bipartite games with strict correlated preferences in the case of
\begin{enumerate}
 \item socially stable matching and locally stable matching,
 \item considerate matching, and
 \item friendship matching for symmetric $\alpha$-values in $\left[0,1\right]$.
\end{enumerate}
In addition, it holds for ordinary bipartite stable matching in the case of
\begin{enumerate}
\setcounter{enumi}{3}
 \item correlated preferences with ties,
 \item strict preferences.
\end{enumerate}
\end{theorem}

\section{General Preferences}
\label{sec:general}

In this section we consider convergence to stable matchings in the two-sided case with general preferences that may be incomplete and have ties. For locally stable matching it is known that in this case there are instances and initial states such that no locally stable matching can be reached using local blocking pair resolutions. Moreover, deciding the existence of a converging sequence of resolutions is \classNP-hard~\cite{HoeferW13}.

We here study the problem for socially, considerate, and friendship matching. Our positive results are based on the following procedure from~\cite{AckermannGMRV11} that is known to construct a sequence of polynomial length for unconstrained stable matching. The only modification of the algorithm for the respective scenarios is to resolve ``social'', ``considerate'' or ``perceived blocking pairs'' in both phases.

\begin{description}
\item[\bf Phase 1] Iteratively resolve only blocking pairs involving a matched vertex $w\in W$. Phase 1 ends when for all blocking pairs $\{u,w\}$ we have $w\in W$ unmatched.
\item[\bf Phase 2] Choose an unmatched $w\in W$ that is involved in a blocking pair. Resolve one of the blocking pairs $\{u,w\}$ that is most preferred by $w$. Repeat until there are no blocking pairs.
\end{description}

It is rather straightforward to see that the algorithm can be applied directly to build a sequence for socially stable matching.

\begin{theorem}
\label{thm:sociallyOneSided}
In every bipartite instance of socially stable matching $G=(V=U\dot{\cup}W,E)$ with general preference lists and social network $L$, for every initial matching $M_0$ there is a sequence of polynomially many improvement steps that results in a socially stable matching. The sequence can be computed in polynomial time.
\end{theorem}

For extended settings the algorithm still works for somewhat restricted social networks. For considerate matching we assume that the link set is only within $L \subseteq (U \times U) \cup (U \times W)$, i.e., no links within partition $W$. 

\begin{theorem}
\label{thm:considerOneSided}
In every bipartite instance of considerate matching $G=(V=U\dot{\cup}W,E)$ with general preference lists and social network $L$ such that $\{w,w'\}\notin L$ for all $w,w'\in W$, for every initial matching $M_0$ there is a sequence of polynomially many improvement steps that results in a considerate matching. The sequence can be computed in polynomial time.
\end{theorem}

We also apply the algorithm to friendship matching in case there can be arbitrary friendship relations $\alpha_{u,u'}, \alpha_{u',u} \ge 0$ for each pair $u,u' \in U$. Here we allow asymmetry with $\alpha_{u,u'} \neq \alpha_{u',u}$. Otherwise, for all $u \in U, w,w'\in W$ we assume that $\alpha_{u,w} = \alpha_{w,u} = \alpha_{w,w'} = 0$, i.e., friendship only exists within $U$.

\begin{theorem}\label{perceivedOnesided}
In every bipartite instance of friendship matching $G=(V=U\dot{\cup}W,E)$ with benefits $b$ and friendship values $\alpha$ such that $\alpha_{u,u'}>0$ only for $u,u'\in U$, for every initial matching $M_0$ there is a sequence of polynomially many improvement steps that results in a friendship matching. The sequence can be computed in polynomial time.
\end{theorem}

The algorithm works fine with links between partitions $U$ and $W$ for the considerate setting, but it fails for positive $\alpha$ between partitions in the friendship case. We defer a discussion to the full version of the paper.  

\subsubsection*{Acknowledgment}
Part of this research was done at the Institute for Mathematical Sciences of NUS Singapore, and at NTU Singapore. The authors thank Edith Elkind for suggesting to study considerate matching.

\bibliographystyle{abbrv}


\begin{thebibliography}{10}

\bibitem{AbrahamLMM08}
D.~Abraham, A.~Levavi, D.~Manlove, and G.~O'Malley.
\newblock The stable roommates problem with globally ranked pairs.
\newblock {\em Internet Math.}, 5(4):493--515, 2008.

\bibitem{AckermannGMRV11}
H.~Ackermann, P.~Goldberg, V.~Mirrokni, H.~R{\"o}glin, and B.~V{\"o}cking.
\newblock Uncoordinated two-sided matching markets.
\newblock {\em SIAM J. Comput.}, 40(1):92--106, 2011.

\bibitem{AnshelevichBH13}
E.~Anshelevich, O.~Bhardwaj, and M.~Hoefer.
\newblock Friendship and stable matching.
\newblock In {\em Proc.\ 21st European Symp.\ Algorithms (ESA)}, pages 49--60,
  2013.

\bibitem{ArcauteV09}
E.~Arcaute and S.~Vassilvitskii.
\newblock Social networks and stable matchings in the job market.
\newblock In {\em Proc.\ 5th Intl.\ Workshop Internet \& Network Economics
  (WINE)}, pages 220--231, 2009.

\bibitem{AskalidisIKMP13}
G.~Askalidis, N.~Immorlica, A.~Kwanashie, D.~Manlove, and E.~Pountourakis.
\newblock Socially stable matchings in the hospitals/residents problem.
\newblock In {\em Proc.\ 13th Workshop Algorithms and Data Structures (WADS)},
  pages 85--96, 2013.

\bibitem{BiroBGKP14}
P.~Bir{\'o}, M.~Bomhoff, P.~A. Golovach, W.~Kern, and D.~Paulusma.
\newblock Solutions for the stable roommates problem with payments.
\newblock {\em Theoret.\ Comput.\ Sci.}, 540:53--61, 2014.

\bibitem{BiroCF08}
P.~Bir{\'o}, K.~Cechl{\'a}rov{\'a}, and T.~Fleiner.
\newblock The dynamics of stable matchings and half-matchings for the stable
  marriage and roommates problems.
\newblock {\em Int.\ J. Game Theory}, 36(3--4):333--352, 2008.

\bibitem{BiroN13}
P.~Bir{\'o} and G.~Norman.
\newblock Analysis of stochastic matching markets.
\newblock {\em Int.\ J. Game Theory}, 42(4):1021--1040, 2013.

\bibitem{BlumRR97}
Y.~Blum, A.~Roth, and U.~Rothblum.
\newblock Vacancy chains and equilibration in senior-level labor markets.
\newblock {\em J. Econom.\ Theory}, 76:362--411, 1997.

\bibitem{BlumR02}
Y.~Blum and U.~Rothblum.
\newblock ``{T}iming is everything'' and martial bliss.
\newblock {\em J. Econom.\ Theory}, 103:429--442, 2002.

\bibitem{BogomolnaiaJ02}
A.~Bogomolnaia and M.~Jackson.
\newblock The stability of hedonic coalition structures.
\newblock {\em Games Econom.\ Behav.}, 38:201--230, 2002.

\bibitem{Cechlarova08}
K.~Cechl{\'a}rova.
\newblock Stable partition problem.
\newblock In {\em Encyclopedia of Algorithms}. 2008.

\bibitem{DiamantoudiMX04}
E.~Diamantoudi, E.~Miyagawa, and L.~Xue.
\newblock Random paths to stability in the roommates problem.
\newblock {\em Games Econom.\ Behav.}, 48(1):18--28, 2004.

\bibitem{Hajdukova06}
J.~Hajdukov{\'a}.
\newblock Coalition formation games: {A} survey.
\newblock {\em Intl.\ Game Theory Rev.}, 8(4):613--641, 2006.

\bibitem{Hoefer11Proc}
M.~Hoefer.
\newblock Local matching dynamics in social networks.
\newblock In {\em Proc.\ 38th Intl.\ Coll.\ Automata, Languages and Programming
  (ICALP)}, volume~2, pages 113--124, 2011.

\bibitem{Hoefer13}
M.~Hoefer.
\newblock Local matching dynamics in social networks.
\newblock {\em Inf.\ Comput.}, 222:20--35, 2013.

\bibitem{HoeferPPSV11}
M.~Hoefer, M.~Penn, M.~Polukarov, A.~Skopalik, and B.~V\"ocking.
\newblock Considerate equilibrium.
\newblock In {\em Proc.\ 22nd Intl.\ Joint Conf.\ Artif.\ Intell.\ (IJCAI)},
  pages 234--239, 2011.

\bibitem{HoeferW13WINE}
M.~Hoefer and L.~Wagner.
\newblock Designing profit shares in matching and coalition formation games.
\newblock In {\em Proc.\ 9th Intl.\ Conf.\ Web and Internet Economics (WINE)},
  pages 249--262, 2013.

\bibitem{HoeferW13}
M.~Hoefer and L.~Wagner.
\newblock Locally stable marriage with strict preferences.
\newblock In {\em Proc.\ 40th Intl.\ Coll.\ Automata, Languages and Programming
  (ICALP)}, volume~2, pages 620--631, 2013.

\bibitem{HoffmanMP13}
M.~Hoffman, D.~Moeller, and R.~Paturi.
\newblock Jealousy graphs: {S}tructure and complexity of decentralized stable
  matching.
\newblock In {\em Proc.\ 9th Intl.\ Conf.\ Web and Internet Economics (WINE)},
  pages 263--276, 2013.

\bibitem{InarraLM08}
E.~Inarra, C.~Larrea, and E.~Moris.
\newblock Random paths to {$P$}-stability in the roommates problem.
\newblock {\em Int.\ J. Game Theory}, 36(3--4):461--471, 2008.

\bibitem{InarraLM10}
E.~Inarra, C.~Larrea, and E.~Moris.
\newblock The stability of the roommate problem revisited.
\newblock Core Discussion Paper 2010/7, 2010.

\bibitem{Irving85}
R.~Irving.
\newblock An efficient algorithm for the "stable roommates" problem.
\newblock {\em J. Algorithms}, 6(4):577--595, 1985.

\bibitem{KlausFW10}
B.~Klaus, F.~Klijn, and M.~Walzl.
\newblock Stochastic stability for roommate markets.
\newblock {\em J. Econom.\ Theory}, 145:2218--2240, 2010.

\bibitem{Knuth76}
D.~Knuth.
\newblock {\em Marriages stables et leurs relations avec d'autres problemes
  combinatoires}.
\newblock Les Presses de l'Universit{\'e} de Montr{\'e}al, 1976.

\bibitem{Manlove13}
D.~Manlove.
\newblock {\em Algorithmics of Matching Under Preferences}.
\newblock World Scientific, 2013.

\bibitem{Mathieu10}
F.~Mathieu.
\newblock Acyclic preference-based systems.
\newblock In X.~Shen, H.~Yu, J.~Buford, and M.~Akon, editors, {\em Handbook of
  peer-to-peer networking}. Springer Verlag, 2010.

\bibitem{RothVV90}
A.~Roth and J.~V. Vate.
\newblock Random paths to stability in two-sided matching.
\newblock {\em Econometrica}, 58(6):1475--1480, 1990.

\bibitem{Woeginger13}
G.~Woeginger.
\newblock Core stability in hedonic coalition formation.
\newblock In {\em Proc.\ 39th Intl.\ Conf.\ Current Trends in Theory \&
  Practice of Comput.\ Sci.\ (SOFSEM)}, pages 33--50, 2013.

\end{thebibliography}

\clearpage
\appendix

\section{Omitted Proofs}

\subsection{Proof of Proposition~\ref{l>2}}
In~\cite[Theorem 3]{Hoefer13} we have shown that such instances and starting states exist for locally stable matching, when agents are allowed to match with partners at a hop distance of at most $\ell = 3$ in $(V, L \cup M)$. This scenario can be embedded into the context of coalition formation games with constraints, where we violate only the above mentioned precondition in the generation rules. Note that the violation is minimal in the sense that we increase from one to at most two sets in the precondition. 

Given an instance of locally stable matching with graph $G=(V,E)$, (social) links $L$, correlated preferences based on edge benefits $b(e)$, we define the parameters of the framework as follows. The set of agents is $N=V$, the set of possible coalitions is $\mathcal{C}=E$. The coalitions that can always be generated are the ones connected by at most 3 links, i.e., $\mathcal{C}_g= E \cap \{ \{u,v\}\mid dist_{L}(u,v)\leq 3\}$. The benefit or coalition weight is obviously given by $w=b$. For the generation rules, we have $T = T_1 \cup T_2 \cup T_3$, where
\begin{align*}
T_1 &=\{(\{\{u,v\}\},\{u,v'\})\mid \{u,v\},\{u,v'\}\in E, \{v,v'\}\in L\}\\
T_2 &=\{(\{\{u,v\}\},\{u',v'\})\mid \{u,v\},\{u',v'\}\in E, \{u',v\},\{u,v'\} \in L\}\\
T_3 &= \{(\{\{u',v\},\{u,v'\}\},\{u,v\})\mid \{u,v\},\{u,v'\},\{u',v\}\in E, \{u',v'\} \in L\}
\end{align*}
Here $T_1$ captures accessible pairs with 2 hops composed of one matching edge and one link, $T_2$ captures accessible pairs within distance of 3 hops composed of one matching edge and two links, and $T_3$ captures accessible pairs within distance of 3 hops composed of two matching edges and one link. The latter generation rules are obviously using two coalitions as precondition. The domination rules implement only the necessary preference-based improvement of coalitions $D = D_w$. 
\qed

\subsection{Proof of Proposition~\ref{generationInconsistent}}
We will attach a sequence of gadgets that imply a unique exponential improvement sequence. We use a gadget of size $9$ and a starting state with the property that to create coalition $C_{6,i}$ of gadget $i$ we twice need to generate $C_{1,i}$. Further the gadget will not reach a stable state unless $C_{6,i}$ exists. Using this property we will connect $k$ such gadgets by allowing a creation rule $\{\{C_{1,i+1}\},C_{6,i}\}$ and identifying $0_{i+1}$ with $8_i$. Then to create $C_{k,6}$ (without which the graph would not be stable) $C_{1,1}$ has to be created at least $2^k$ times.\\
Now for gadget $i$ we have $N_i=\{0_i,\ldots, 8_i\}$, $\mathcal{C}_i=\{C_{1,i},\ldots, C_{6,i}\}$ with $C_{1,i}=\{0_i,1_i,2_i\}$, $C_{2,i}=\{1_i,3_i\}$, $C_{3,i}=\{3_i,4_i,5_i\}$, $C_{4,i}=\{4_i,6_i\}$, $C_{5,i}=\{2_i,6_i,7_i\}$, $C_{6,i}=\{5_i,7_i,8_i\}$, weights $w(C_{1,i})=x_i+1$, $w(C_{2,i})=x_i+2$, $w(C_{3,i})=x_i+4$, $w(C_{4,i})=x_i+3$, $w(C_{5,i})=x_i+2$ and $w(C_{6,i})=x_i+5$ with $x_i=5(i-1)$, and generation rules
\begin{align*}
T_1 &=\{\{\{C_{1,1}\},C_{2,1}\},\{\{C_{1,1}\},C_{5,i}\},\{\{C_{2,1}\},C_{3,i}\},\{\{C_{3,1}\},C_{1,1}\},\{\{C_{4,1}\},C_{1,1}\},\\
& \hspace{0.66cm} \{\{C_{5,1}\},C_{6,1}\}\}\enspace,\\
T_i &=\{\{\{C_{1,i}\},C_{2,i}\},\{\{C_{1,i}\},C_{5,i}\},\{\{C_{2,i}\},C_{3,i}\},\{\{C_{3,i}\},C_{4,i-1}\},\\
& \hspace{0.66cm}\{\{C_{4,i}\},C_{4,i-1}\}, \{\{C_{5,i}\},C_{6,i}\}\}~~\mbox{ if }i>1\enspace.
\end{align*}
The set $D$ of domination rules is empty except for all rules of the form $(\{C\},C')$ such that $w(C) \geq w(C')$ and $C \cap C' \neq \emptyset$. As starting coalition structure we have $\{C_{4,k}\}$.\\

The dynamics are best understood when using the object movement graph instead of dealing with the single vertices. In Figure~\ref{bildExpGeneration} we give the object movement graph of the first two gadgets to visualize the dynamics of the gadgets themselves as well as their interaction. We will analyze the dynamics of gadget $1$. The subsequent gadgets work similarly. In the beginning there are no coalitions in gadget 1 so we first have a look at how to get some coalition to start from. Now every $C_{4,i}$ for $i>1$ can only be used to generate $C_{4,i-1}$. Note that in this case $C_{4,i}$ is not deleted. Thus in the beginning those $C_{4,i}$ are one by one created until we reach $C_{4,1}$. With $C_{4,1}$ we can only generate $C_{1,1}$ (in gadgets $i>1$ in this situation we might generate $C_{4,i-1}$ and then ''wait'' for $C_{1,i}$ to arrive). Next as $C_{5,1}$ is blocked by $C_{4,1}$ the only option is to generate $C_{2,1}$ and thus lose $C_{1,1}$ again. From there we can only generate $C_{3,1}$ while losing $C_{2,1}$ and $C_{4,1}$. With the remaining coalition $C_{3,1}$ we can recover $C_{1,1}$ which now leads to creating $C_{5,1}$ and losing $C_{1,1}$ a second time. Next we can only create $C_{6,1}$ which causes the deletion of $C_{3,1}$ and $C_{5,1}$. Finally $C_{6,1}$ can now be used to create $C_{1,2}$ which leaves gadget 1 empty. But at the latest after $C_{3,2}$ has been created in the next step $C_{4,1}$ has to be created again and we can rerun the dynamics for the gadget in the same manner.
\qed

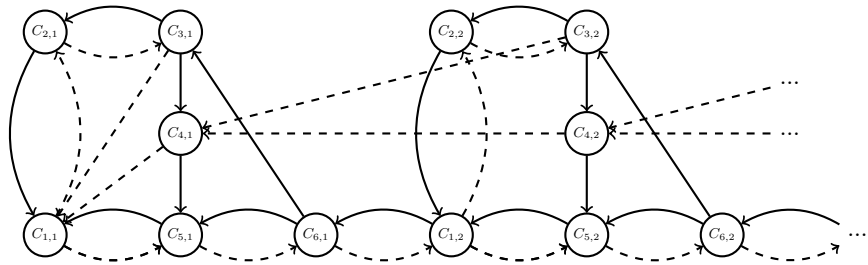
\begin{figure}
\newcommand{\sca}{0.6}
\begin{center}
\begin{tikzpicture}[thick,scale=0.45]
\tikzstyle{leer} = [draw=none,fill=none]

\node[scale=\sca] (11) at (0,3)[draw=black, circle]{$C_{1,1}$};
\node[scale=\sca] (21) at (0,9)[draw=black, circle]{$C_{2,1}$};
\node[scale=\sca] (31) at (4,9)[draw=black, circle]{$C_{3,1}$};
\node[scale=\sca] (41) at (4,6)[draw=black, circle]{$C_{4,1}$};
\node[scale=\sca] (51) at (4,3)[draw=black, circle]{$C_{5,1}$};
\node[scale=\sca] (61) at (8,3)[draw=black, circle]{$C_{6,1}$};

\path[->]
(21) edge [bend right] (11)
(31) edge [bend right] (21)
(51) edge [bend right] (11)
(61) edge [bend right] (51)
(61) edge  (31)
(31) edge (41)
(41) edge (51);

\path[->, dashed]
(11) edge [bend right] (21)
(21) edge [bend right] (31)
(31) edge (11)
(11) edge [bend right] (51)
(11) edge [bend right] (51)
(51) edge [bend right] (61)
(41) edge (11);
%
\node[scale=\sca] (12) at (12,3)[draw=black, circle]{$C_{1,2}$};
\node[scale=\sca] (22) at (12,9)[draw=black, circle]{$C_{2,2}$};
\node[scale=\sca] (32) at (16,9)[draw=black, circle]{$C_{3,2}$};
\node[scale=\sca] (42) at (16,6)[draw=black, circle]{$C_{4,2}$};
\node[scale=\sca] (52) at (16,3)[draw=black, circle]{$C_{5,2}$};
\node[scale=\sca] (62) at (20,3)[draw=black, circle]{$C_{6,2}$};

\path[->]
(22) edge [bend right] (12)
(32) edge [bend right] (22)
(52) edge [bend right] (12)
(62) edge [bend right] (52)
(62) edge  (32)
(32) edge (42)
(42) edge (52)
(12) edge [bend right] (61);

\path[->, dashed]
(12) edge [bend right] (22)
(22) edge [bend right] (32)
(32) edge (41)
(12) edge [bend right] (52)
(12) edge [bend right] (52)
(52) edge [bend right] (62)
(42) edge (41)
(61) edge [bend right] (12);

\node (13) at (24,3) [draw=none]{...};
\node (33) at (22,7.5) [draw=none]{...};
\node (43) at (22,6) [draw=none]{...};

\path[->]
(13) edge [bend right] (62);

\path[->, dashed]
(33) edge (42)
(43) edge (42)
(62) edge [bend right] (13);

\end{tikzpicture}
\end{center}
\caption{\textbf{Object movement graph of gadget 1 and 2:} the thick edges symbolize domination rules and the dashed edges symbolize generation rules}\label{bildExpGeneration}
\end{figure}

%

\subsection{Proof of Lemma~\ref{embedding}}

We will consider each setting individually. Further while we describe the embedding for each case we will only give a proof of correctness for friendship matching with $k=1$ as the proofs for all cases are very similar and this is one of the most complex ones. 

\subsubsection{Locally Stable Matching}
To embed an instance of locally stable matching given by graph $G$, link set $L$, and edge benefits into our framework, we apply the embedding indicated in Proposition~\ref{l>2}. Note that for the standard case of $k=1$ matching edges and lookahead $\ell = 2$, we obtain consistent generation and domination rules. If we change to $k > 1$ matching edges or lookahead of $\ell > 2$, consistency of generation rules becomes violated, as two edges can appear in the precondition. This follows from the exponential lower bounds in~\cite{Hoefer13,HoeferW13}.

\subsubsection{Socially Stable Matching}
To embed an instance of socially stable matching given by graph $G=(V,E)$, link set $L$, and edge benefits $b$ into our framework, we use $N=V$, $\mathcal{C}=E$, and $\mathcal{C}_g= L \cap E$. There are no additional generation rules $T=\emptyset$, and the benefits are obviously given by $w=b$. The domination rules implement only the necessary preference-based improvement of coalitions $D = D_w$. Obviously, generation and domination rules are consistent. 

If we change to $k > 1$ matching edges per agent, then, in principle, we violate a condition of our framework -- that in every state, every agent can be part of at most one coalition. A simple way to embed the games correctly into the framework is to represent each agent $u$ by $k$ auxiliary agents $u_1,\ldots,u_k$, who can match to one partner each. The edges between two agents become a complete bipartite graph between the corresponding auxiliary agents. Note that in the new game, $u_i, u_j$ and $v_{i'},v_{j'}$ can now build two edges among $u$ and $v$. This can be prohibited using an auxiliary agent $a_{u,v}$ for each matching edge $\{u,v\} \in E$, and replacing every coalition $\{u_i,v_j\}$ by $\{u_i,v_j,a_{u,v}\}$, for $i,j = 1,\ldots, k$. 

More formally, we define 
\begin{align*}
N &=\{v_i\mid v\in V, ~i=1\ldots k\}\cup \{ a_{u,v} \mid \{u,v\} \in E\}\\
\mathcal{C}&= \{\{u_i,v_j,a_{u,v}\}\mid \{u,v\}\in E, ~i,j=1\ldots k\}\\
\mathcal{C}_g &= \{\{u_i,v_j,a_{u,v}\}\mid \{u,v\}\in E \cap L, ~i,j=1\ldots k\}\\
w(\{u_i,v_j,a_{u,v}\})&=b(u,v) \text{ for } \{u,v\}\in E, ~i,j=1\ldots k\\
T &= \emptyset\\
D &=D_w
\end{align*}
Obviously, generation and domination rules are consistent. 

\subsubsection{Considerate Matching}
To embed an instance of friendship matching given by graph $G=(V,E)$, link set $L$, and edge benefits $b$, we use $N=V$, $\mathcal{C}=\mathcal{C}_g=E$, and, consequently, $T=\emptyset$. We set $w=b$. For the domination rules, let $D=D_1\cup D_w$ where
\begin{align*}
 D_1 &= \{(\{\{u,v\}\},\{u,v'\}) \mid \{u,v\},\{u,v'\} \in E, \{u,v\}\in L \text{ or } \{v,v'\}\in L\}  \enspace.
\end{align*}
The case of $k>1$ matching edges can be handled via auxiliary agents as explained for socially stable matching above. Again, this essentially affects only the domination rules, which allow more flexibility with respect to consistency. More formally, we define $N$, $\mathcal{C}$, $w$ and $T$ as outlined above and change the remaining definitions to
\begin{align*}
\mathcal{C}_g &= \mathcal{C} \\
D &= D_1 \cup D_w \text{ with }\\
D_1 &= \{(\{\{u_i,v_j,e_{u,v}\}\},\{u_{i'},v'_{j'},e_{u,v}\})\mid \\& \hspace{0.66cm} \{u,v\},\{u,v'\}\in E, \{u,v\}\in L \text{ or } \{v,v'\}\in L,~i,i',j,j'=1\ldots k\} \enspace.
\end{align*}

\subsubsection{Friendship Matching}
To embed an instance of friendship matching given by graph $G=(V,E)$, edge benefits $b$, symmetric friendship values $\alpha$, we use $N=V$, $\mathcal{C}=\mathcal{C}_g=E$, and, consequently, $T=\emptyset$. To model the perceived utilities, we assume $w(\{u,v\})=b(\{u,v\})+\alpha_{u,v}b(\{u,v\})$. For the domination rules, let $D=D_1\cup D_2\cup D_w$ where
\begin{align*}
 D_1 &= \{(\{\{u,v\}\},\{u,v'\})\mid \\
 & \hspace{0.66cm} \alpha_{v',v}b(\{u,v\})+\alpha_{v',u}b(\{u,v\})\geq b(\{u,v'\})+\alpha_{v',u}b(\{u,v'\})\} \\
 D_2 &= \{(\{\{u,v\},\{u',v'\}\},\{u,v'\})\mid \\
 & \hspace{0.66cm} b(\{u',v'\})+\alpha_{v',u'}b(\{u',v'\})+\alpha_{v',v}b(\{u,v\})+\alpha_{v',u}b(\{u,v\})\\  &\hspace{0.76cm}\geq b(\{u,v'\})+\alpha_{v',u}b(\{u,v'\})\} \enspace.
\end{align*}
The domination rules in $D_1$ describe that $v'$ earns more from existing $\{u,v\}$ through friendship than from the candidate $\{u,v'\}$. In $D_2$, agent $v'$ earns more from the combination of benefits from dropped agents than from the candidate $\{u,v'\}$. Again, the case of $k>1$ edges per agent can be included using auxiliary agents and preserves consistency as it affects only the domination rules.

Regarding correctness we only prove the case of $k=1$. The other cases are very similar.

Let $M$ be a matching in $G$. Now assume we have a perceived blocking pair $\{u,v'\}$ for $M$ that we intend to resolve. Domination can only occur through edges involving $u$ or $v'$. 

Firstly, if $u$ and $v'$ are unmatched, there cannot be any edges dominating $\{u,v'\}$, and we can generate all matching edges as candidate coalitions via $\mathcal{C}_g=\mathcal{C}$. Hence, perceived blocking pairs between unmatched agents are also undominated candidate coalitions. After adding $\{u,v'\}$, no edge is removed. Hence, the set of coalitions resulting from the rules above exactly represents the matching after resolving the perceived blocking pair $\{u,v'\}$.

Secondly, assume that agent $u$ is matched to some agent $v$, but agent $v'$ is unmatched. As $\{u,v'\}$ is a perceived blocking pair, we know that $u$ improves by switching from $v$ to $v'$, that is, 
\begin{align*}
 &  b(\{u,v\})+\alpha_{u,v} b(\{u,v\})+\sum_{u'\in V\setminus\{u,v,v'\}}\alpha_{u,u'}b(M,u')\\
 &< \;\; b(\{u,v'\})+\alpha_{u,v'} b(\{u,v'\})+\sum_{u'\in V\setminus\{u,v,v'\}}\alpha_{u,u'}b(M,u')
\end{align*}
which cancels to $b(\{u,v\})+\alpha_{u,v} b(\{u,v\})< b(\{u,v'\})+\alpha_{u,v'} b(\{u,v'\})$. Thus $\{u,v'\}$ is not dominated by $\{u,v\}$ through $D_w$. Now $\{u,v\}$ might still dominate $\{u,v'\}$ through $D_1$. But then $\alpha_{v,v'}b(\{u,v\})+\alpha_{u,v'}b(\{u,v\})\geq b(\{u,v'\})+\alpha_{u,v'}b(\{u,v'\})$, that is, the gain $v'$ receives through its friendships with $v$ and $u$ from $\{u,v\}$ is at least as large as the gain it would receive by matching with $u$ (directly and through friendship). This contradicts the assumption that $\{u,v'\}$ is a perceived blocking pair. Hence $\{u,v'\}$ is an undominated candidate coalition. After adding $\{u,v'\}$, $\{u,v\}$ is dominated through weight and hence gets dropped. Again, the set of coalitions resulting from our rules exactly correspond to the matching after resolving the perceived blocking pair $\{u,v'\}$.

Thirdly, assume that $\{u,v\}$ and $\{u',v'\}$ are present in $M$. The previous arguments for edges that dominate $\{u,v'\}$ through $D_1$ or $D_w$ can be applied again. It remains to check whether domination via $D_2$ is possible. But the domination rules in $D_2$ imply that the loss caused by giving up $\{u,v\}$ \textbf{and} $\{v',u'\}$ for $v'$ is at least as large as the gain generated from $\{u,v'\}$. Thus, as $\{u,v'\}$ is a perceived blocking pair, there is no rule in $D_2$ pointing from $\{\{u,v\},\{u',v'\}\}$ to $\{u,v'\}$ and we can simply generate $\{u,v'\}$ again. Then $\{u,v\}$ and $\{u',v'\}$ are dominated via $D_w$ and hence get dropped, which gives $M\setminus\{\{u,v\},\{u',v'\}\}\cup\{\{u,v'\}\}$. Again, the set of coalitions resulting from our rules exactly correspond to the matching after resolving the perceived blocking pair $\{u,v'\}$.

Conversely, let $\mathcal{S}$ be a feasible coalition structure in our coalition formation game with constraints. Observe that our rules imply that $\mathcal{S}$ corresponds to a matching $M$. Further let $\{u,v'\} \not\in \mathcal{S}$ be an unmarked, undominated edge. Assume for contradiction that $\{u,v'\}$ is not a perceived blocking pair because $v'$ would not improve. 

Firstly, if $v'$ is single, this is only possible if $v'$ gains at least as much through $u$'s current matching edge than through directly matching to $u$. Then, a rule of $D_1$ would dominate $\{u,v'\}$. 

Secondly, assume that $v'$ is matched to some $u'$ but $u$ is unmatched. Then $\{u,v'\}$ is not a perceived blocking pair if $b(\{u',v'\})+\alpha_{u',v'}b(\{u',v'\})\geq b(\{u,v'\})+\alpha_{u,v'}b(\{u,v'\})$. But in that case $\{u',v'\}$ would dominate $\{u,v'\}$ by $D_w$. 

Thirdly, let $v'$ be matched to some $u'$ and $u$ matched to some $v$. As before we must have $b(\{u',v'\})+\alpha_{u',v'}b(\{u',v'\})< b(\{u,v'\})+\alpha_{u,v'}b(\{u,v'\})$ and $b(\{u,v\})+\alpha_{u,v}b(\{u,v\})< b(\{u,v'\})+\alpha_{u,v'}b(\{u,v'\})$ as otherwise $\{u,v'\}$ would be dominated by $D_w$. $\{u,v'\}$ is no improvement for $v'$ only if the combined loss caused by $\{u,v\}$ and $\{u',v'\}$ out-weights the gain through $\{u,v'\}$. In other words, 
\begin{align*}
& b(\{u',v'\})+\alpha_{u',v'}b(\{u',v'\})+\alpha_{v,v'}b(\{u,v\})+\alpha_{u,v'}b(\{u,v\})\\
\geq \; & b(\{u,v'\})+\alpha_{u,v'}b(\{u,v'\})\enspace.
\end{align*} 
Then there is an according rule in $D_2$ and the combined existence of $\{u',v'\}$ and $\{u,v\}$ results in $\{u,v'\}$ being dominated. 

Hence, whenever $\{u,v'\}$ is an undominated candidate coalition, it represents a perceived blocking pair for the current matching. Further when it gets inserted, any former matching edges of $u$ and $v'$ get dominated by weight and dropped, while all other edges remain unaffected. Hence, the new coalition structure represents exactly matching $M$ after resolving $\{u,v'\}$.
\qed

\subsection{Proof of Theorem~\ref{alwaysShortSequence}}
The proof generalizes a similar result for locally stable matchings~\cite{HoeferW13} and is based on two observations:
\begin{enumerate}
\item Note that by design of $T$ and $D$, if some generation rule $(\{C_1\},C_2)$ is finally applied, we need to have $w(C_2)>w(C_1)$ and the creation of $C_2$ causes the deletion of $C_1$. Thus, within any sequence of improvement steps we can identify a unique predecessor for each coalition $C \notin\mathcal{C}_g$ whose presence is necessary for creation of $C$. Furthermore, this predecessor has weight strictly smaller than $w(C)$. Hence the sequence of predecessors necessary to generate any coalition $C$ is limited by $m$.
\item Our second observation is that the only domination rules that are applied in the deletion of an existing coalition are those based on weight domination. Thus, every deletion of a coalition is accompanied by the creation of a worthier coalition. A chain of deletions again is limited in length by $m$.
\end{enumerate}

Now let $\mathcal{I}$ be some sequence of improvement steps converting $\mathcal{S}_0$ into $\mathcal{S}^*$. If a coalition is created and deleted again but neither used to delete another coalition nor marked as a predecessor to create one, then this coalition provides no contribution for the transformation from $\mathcal{S}_0$ to $\mathcal{S}^*$. Thus we can delete all these coalitions from $\mathcal{I}$ and receive a sequence $\mathcal{I}_1$ which as well transforms $\mathcal{S}_0$ into $\mathcal{S}^*$ via improvement steps. Now in $\mathcal{I}_1$ there might be coalitions which get created and deleted again without use as they only deleted or created coalitions we dropped from $\mathcal{I}_1$. Thus we can repeat this sequence truncation until all remaining coalitions are of use. We claim that this sequence $\mathcal{I}^*$ has to be of polynomial length.

First assume state $\mathcal{S}_0$ is the empty coalition structure. Then we do not have to delete any unfitting coalitions but simply create the needed ones. As not all of $\mathcal{S}^*$ might be in $\mathcal{C}_g$ we possibly have to use generation rules of $T$ but by (1) we know that for each desired coalition we need at most $m$ steps. Thus, overall we need at most $m^2$ steps. Now for an arbitrary starting coalition structure we might also have to delete certain coalitions to reach $\mathcal{S}*$. Thus each of coalition of $\mathcal{S}_0$ might generate a chain of coalitions deleting each other throughout the sequence, but (2) tells us that this chain is limited in length by $m$. Also, (1) again tells us that the number of steps it takes to generate the coalition which is used for the deletion is limited by $m$ as well. The only remaining issue is to argue why additional deletion of coalitions during this procedure does not create problems. Now if such a coalition was part of $\mathcal{S}_0$ it does not create any additional costs. If it was part of some deletion chain, its cost was already accounted towards the coalition of $\mathcal{S}_0$ which had to be deleted. In all other cases, the creation of this coalition was of no use in the first place, that is, it would not be part of the truncated sequence $\mathcal{I}^*$. Hence, overall we have a sequence of length at most $|\mathcal{S}_0|\cdot m^2 + |\mathcal{S}^*|\cdot m \in O(m^2n)$ steps. \qed 

\subsection{Proof of Theorem~\ref{np}}
The proof is done via reduction from \ThreeSAT\ . We will use the same idea and central construction for all cases and only adapt the structure of the clause-gadgets to the specific settings. Each clause gadget will have the property that one particular vertex $x_C$ has be matched to a vertex of the central construction at some (arbitrary) point during the dynamics and has to be left single again. Otherwise the clause gadget cannot be transformed into the state it has in the desired final matching.

We first outline the universal proof approach including only this one particular vertex $x_C$ per clause $C$ (and the central construction). We show that it is \classNP-hard to decide whether there is a sequence of improvement steps such that each of the clause vertices gets matched and dropped at least once. Afterwards, for every setting we will give the exact clause gadget and explain why it is necessary to match $x_C$ to some vertex outside the clause gadget to reach the final state.

Given a \ThreeSAT\ formula with $k$ variables $x_1,\ldots,x_k$ and $l$ clauses $C_1,\ldots,C_l$, where clause $C_j$ contains the literals $l_{1,j}, l_{2,j}$ and $l_{3,j}$, for our central construction we have
\begin{eqnarray*}
    U =&&\{u_{x_i}|i=1\ldots k\} \cup \{u_{\overline{x}_i}|i=1\ldots k\} \cup \{x_{C_j}|j=1\ldots l\},\\
    W =&&\{w_{x_i}|i=1\ldots k\} \cup \{w_{\overline{x}_i}|i=1\ldots k\}.
\end{eqnarray*}
Further $E=E_1\cup E_2\cup E_3$ with
$E_1=\{u_{x_i},w_{x_i}\}, \{u_{\overline{x}_i},w_{\overline{x}_i}\}\mid i=1\ldots k\}$, 
$E_2=\{u_{x_i},w_{\overline{x}_i}\}, \{u_{\overline{x}_i},w_{x_i}\}\mid i=1\ldots k\}$, 
and $E_3=\{\{x_{C_j},w_{l_{i,j}}\}\mid j=1\ldots l, i=1\ldots 3\}$, 
and benefits in Table~\ref{benefits}.

\setlength{\tabcolsep}{5pt}
\begin{center}
\begin{table}
 \caption{Table of edge benefits}
\centering
\begin{tabular}{|l|c|c|l|}\hline
$U$ & $W$ & $b(\{u,w\})$ & ~\\\hline
$x_{C_j}$ & $w_{l_{i,j}}$ & $i\cdot l+j$ & $j=1\ldots l$, $i=1\ldots 3$\\\hline
$u_{x_i}$ & $w_{\overline{x}_i}$ & $4l+i$ & $i=1\ldots k$\\\hline
$u_{\overline{x}_i}$ & $w_{x_i}$ & $4l+k+i$ & $i=1\ldots k$\\\hline
$u_{x_i}$ & $w_{x_i}$ & $4l+2k+i$ & $i=1\ldots k$\\\hline
$u_{\overline{x}_i}$ & $w_{\overline{x}_i}$ & $4l+3k+i$ & $i=1\ldots k$.\\\hline
\end{tabular}
       \label{benefits}
\end{table}
\end{center}
For a small example see Figure~\ref{bild3Sat}. In the case of locally and socially stable matching we will have social links between all vertices of $U$ and $W$ to make sure that all edges of $E$ are available for matching at all times. In the case of friendship matching we set all $\alpha$ to $0$ to ensure that utility is also perceived utility.

We start from $M_0=E_2$ and want to reach $M^*=E_1$ which also is the only stable state of this graph. This transformation is always possible, but we now want to decide whether we can construct a sequence which involves all vertices $x_{C_j}$.

First, let us build an intuition what has to happen to match each $x_{C_j}$. Note that we have to create some edge $\{x_{C_j},w_{l_{i,j}}\}$ of $E_3$ for every clause $C_j$, but in the beginning all those edges are blocked through $E_2$. During the dynamics per variable we can switch one edge of $E_2$ to $E_1$ freeing the other $w$-vertex. Then this vertex can be used to ``visit'' all the adjacent clauses in increasing order before creating the second edge of $E_1$. But the $w$-vertex which switched first remains blocked and thus can be used for none of the clauses. Thus, the choice whether to create $\{u_{x_i},w_{x_i}\}$ or $\{u_{\overline{x}_i},w_{\overline{x}_i}\}$ first can be seen as the choice whether to set $x_i$ true or false (by creating the opposite edge first). All clauses that include the variable in the corresponding assignment then can be matched using $w_{x_i}$ respectively $w_{\overline{x}_i}$. We will now formally prove the correctness of the reduction.

Assume that the \ThreeSAT\ formula is satisfiable. Then we pick a satisfying assignment and for each variable generate the edge of $E_1$ which symbolizes the inverses of the assignment. Now the $w$-vertex in the assigned value of every variable is unmatched and we one by one generate the incident edges leading to the clause variables in increasing order starting from the smallest unblocked edge. As for every clause at least one literal is satisfied and the edges are created in increasing order and thus cannot block each other, by the end of this phase all vertices $x_{C_j}$ were matched at least once. It remains to generate the second edge for every variable, and we have reached $M^*$ with a sequence of the desired form.

Assume that we can reach $M^*$ from $M_0$ with a sequence matching each $x_{C_j}$ at least once. For each clause $C_j$ pick a vertex $w_{l_{i,j}}$ which was matched to $x_{C_j}$. We claim that for no variable $x_{i}$ both vertices $w_{x_i}$ and $w_{\overline{x}_i}$ are picked: In the beginning both vertices are matched through an edge larger than any edge leading to a clause vertex. Thus to match one of these vertices to some $x_{C_j}$ it first has to become single, that is, its matching partner $u_{\overline{x}_i}$ respectively $u_{x_i}$ has to leave for a better partner. But the only better partner for $u_{\overline{x}_i}$ is $w_{\overline{x}_i}$ and the only better partner for $u_{x_i}$ is $w_{x_i}$. Further, both edges then are stable as they are the top choice of both partners. Hence, to make $w_{x_i}$ available we have to block $w_{\overline{x}_i}$ for the rest of the dynamics and to make $w_{\overline{x}_i}$ available we have to block $w_{x_i}$ for the rest of the dynamics. Now as at most one $w$-vertex of each variable is picked, we can assign each variable the value of the picked vertex and further assign a random value to each variable with no $w$-vertex picked. Then for each clause at least one literal is fulfilled, that is, the formula is satisfied.
\newcommand{\scale}{0.7}
\begin{figure}[ht]
\begin{center}
\begin{tikzpicture}[thick,scale=0.68]
\tikzstyle{leer} = [draw=none,fill=none]

\node[scale=\scale] (ux) at (0,7)[draw=black, circle, fill=none]{$u_{x}$};
\node[scale=\scale] (unx) at (3,7)[draw=black, circle, fill=none]{$u_{\overline{x}}$};
\node[scale=\scale] (wx) at (0,4)[draw=black, circle, fill=none]{$w_{x}$};
\node[scale=\scale] (wnx) at (3,4)[draw=black, circle, fill=none]{$w_{\overline{x}}$};

\node[scale=\scale] (uy) at (6,7)[draw=black, circle, fill=none]{$u_{y}$};
\node[scale=\scale] (uny) at (9,7)[draw=black, circle, fill=none]{$u_{\overline{y}}$};
\node[scale=\scale] (wy) at (6,4)[draw=black, circle, fill=none]{$w_{y}$};
\node[scale=\scale] (wny) at (9,4)[draw=black, circle, fill=none]{$w_{\overline{y}}$};

\node[scale=\scale] (uz) at (12,7)[draw=black, circle, fill=none]{$u_{z}$};
\node[scale=\scale] (unz) at (15,7)[draw=black, circle, fill=none]{$u_{\overline{z}}$};
\node[scale=\scale] (wz) at (12,4)[draw=black, circle, fill=none]{$w_{z}$};
\node[scale=\scale] (wnz) at (15,4)[draw=black, circle, fill=none]{$w_{\overline{z}}$};

\node[scale=\scale] (cj) at (7.5,2) [draw=black, circle, fill=none]{$x_{C_j}$};

\path[-,dashed]
(ux) edge node[scale=0.9,left]{\tiny $w_x+2k$} (wx)
(ux) edge node[scale=0.9,above left]{\tiny $w_x~~~$} (wnx)
(unx) edge node[scale=0.9,above right]{\tiny $~~~w_x+k$} (wx)
(unx) edge node[scale=0.9,right]{\tiny $w_x+3k$} (wnx)
(uy) edge node[scale=0.9,left]{\tiny $w_y+2k$} (wy)
(uy) edge node[scale=0.9,above left]{\tiny $w_y~~~$} (wny)
(uny) edge node[scale=0.9,above right]{\tiny $~~~w_y+k$} (wy)
(uny) edge node[scale=0.9,right]{\tiny $w_y+3k$} (wny)
(uz) edge node[scale=0.9,left]{\tiny $w_z+2k$} (wz)
(uz) edge node[scale=0.9,above left]{\tiny $w_z~~~$} (wnz)
(unz) edge node[scale=0.9,above right]{\tiny $~~~w_z+k$} (wz)
(unz) edge node[scale=0.9,right]{\tiny $w_z+3k$} (wnz)
(cj) edge node[scale=0.9,below]{\tiny $l+j$} (wx)
(cj) edge node[scale=0.9,left]{\tiny $2l+j$~~} (wny)
(cj) edge node[scale=0.9,below]{\tiny $~~3l+j$} (wz)
;

\end{tikzpicture}
\end{center}
\caption{Central gadget with variables $x$, $y$, $z$ and clause $C_j=x\vee \overline{y} \vee z$}\label{bild3Sat}
\end{figure}
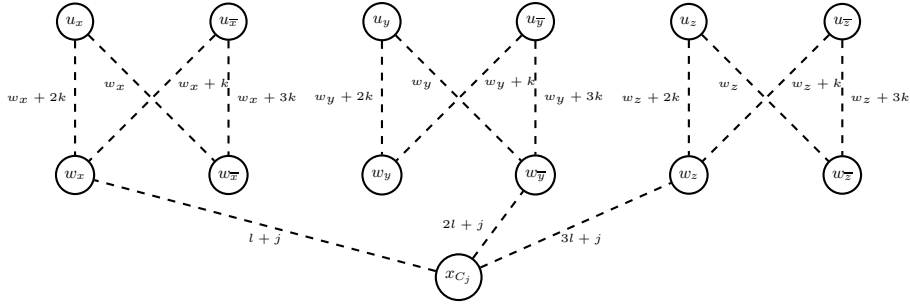
Finally, we design appropriate clause gadgets for each case:
\begin{enumerate}
\item For \emph{socially and locally stable matching} we add a vertex $y_{C_j}$ to $W$ and an edge $\{x_{C_j},y_{C_j}\}$ of benefit $j$ to $E$ for every clause $C_j$. Further we also add all the edges $\{x_{C_j},y_{C_j}\}$ to the starting state $M_0$ but keep $M^*$. Note that we did no add any social links for $y_{C_j}$. Thus $M^*$ is stable and can be reached if and only if we rematch every $y_{C_j}$ at least once (and hence delete $\{x_{C_j},y_{C_j}\}$).

\item For \emph{considerate matching} we add two vertices $y_{C_j}$ and $y'_{C_j}$ to $W$ and edges $\{x_{C_j},y_{C_j}\}$ of benefit $j-\frac{1}{2}$ and $\{x_{C_j},y'_{C_j}\}$ of benefit $j$ to $E$ for every clause $C_j$. Further we also add all the edges $\{x_{C_j},y_{C_j}\}$ to the starting state $M_0$ and all the edges $\{x_{C_j},y'_{C_j}\}$ to $M^*$. Finally we introduce a social link between $y_{C_j}$ and $y'_{C_j}$. Now $x_{C_j}$ cannot switch from $y_{C_j}$ to $y'_{C_j}$ as $y'_{C_j}$ is friends with $y_{C_j}$ and would thus reject $x_{C_j}$. But, if $x_{C_j}$ is single, $y'_{C_j}$ does not reject $x_{C_j}$. Hence again we need to make sure that for every clause $x_{C_j}$ is matched to some vertex outside the clause-gadget and dropped to reach $M^*$.

\item For \emph{friendship matching} we add two vertices $y_{C_j}$ and $y'_{C_j}$ to $W$ and edges $\{x_{C_j},y_{C_j}\}$ of benefit $j-\frac{1}{2}$ and $\{x_{C_j},y'_{C_j}\}$ of benefit $j$ to $E$ for every clause $C_j$. The only friendship value $\neq 0$ is $\alpha_{x_{C_j},y_{C_j}}=\frac{1}{2j-1}$. Again we add all the edges $\{x_{C_j},y_{C_j}\}$ to the starting state $M_0$ and all the edges $\{x_{C_j},y'_{C_j}\}$ to $M^*$. Note that by the choice of $\alpha_{x_{C_j},y_{C_j}}$ the perceived value for $x_{C_j}$ from $\{x_{C_j},y_{C_j}\}$ now is $(1+\alpha_{x_{C_j},y_{C_j}})(j-\frac{1}{2})=j-\frac{1}{2}+(j-\frac{1}{2})\frac{1}{2j-1}=j=b(\{x_{C_j},y'_{C_j}\})$, that is, there is a tie in $x_{C_j}$'s preference list regarding $y_{C_j}$ and $y'_{C_j}$. Hence $M^*$ is stable but $x_{C_j}$ will not switch directly from $y_{C_j}$ to $y'_{C_j}$. But once $x_{C_j}$ is single we can match it with $y'_{C_j}$ as desired.
 
\item For \emph{correlated matching with ties} we add two vertices $y_{C_j}$ and $y'_{C_j}$ to $W$ and edges $\{x_{C_j},y_{C_j}\}$ and $\{x_{C_j},y'_{C_j}\}$, both of benefit $j$, to $E$ for every clause $C_j$. Further we also add all the edges $\{x_{C_j},y_{C_j}\}$ to the starting state $M_0$ and all the edges $\{x_{C_j},y'_{C_j}\}$ to $M^*$. Then $x_{C_j}$ does not switch from $y_{C_j}$ to $y'_{C_j}$ as it yields no improvement. But, if $x_{C_j}$ is single, we can choose to match to $y'_{C_j}$.
 
\item For \emph{matching with strict preferences} we first note that, as all edge values in the central gadget are distinct, we can derive a strict preference order over all possible matching partners for each vertex. Now for each clause $C_j$ we add one vertex $x'_{C_j}$ to $U$ and two vertices $y_{C_j}$ and $y'_{C_j}$ to $W$ and edges $\{x_{C_j},y_{C_j}\}$, $\{x_{C_j},y'_{C_j}\}$, $\{x'_{C_j},y_{C_j}\}$ and $\{x'_{C_j},y'_{C_j}\}$ to $E$. For $x_i$ we add $y_{C_j}>_{x_{C_j}}y'_{C_j}$to the bottom of the preference list, that is, all vertices of the central gadget are preferred. For the other preferences we have $y'_{C_j}>_{x'_{C_j}}y_{C_j}$, $x'_{C_j}>_{y_{C_j}}x_{C_j}$ and $x_{C_j}>_{y'_{C_j}}x'_{C_j}$. To $M_0$ we add $\{x_{C_j},y_{C_j}\}$ and $\{x'_{C_j},y'_{C_j}\}$ and to $M^*$ we add $\{x_{C_j},y'_{C_j}\}$ and $\{x'_{C_j},y_{C_j}\}$. Now the clause gadget has two stable states: $\{\{x_{C_j},y_{C_j}\},\{x'_{C_j},y'_{C_j}\}\}$ and  $\{\{x_{C_j},y'_{C_j}\},\{x'_{C_j},y_{C_j}\}\}$. To switch again we first have to break open the stable starting state by matching $x_{C_j}$ to some vertex of the central gadget and then leave $x_{C_j}$ single. Then $y'_{C_j}$ can switch to its preferred choice $x_{C_j}$ which frees $x'_{C_j}$ for $y_{C_j}$ resulting in the desired final state.
\end{enumerate}
\qed

\subsection{Proof of Theorem~\ref{thm:sociallyOneSided}}
The proof is almost identical to the proof for the general case. The only modification is the limitation to pairs that represent social links for the rematching process.
 
In Phase 1 each matched $w\in W$ increases in terms of utility (or becomes unmatched) and the number of matched $w$ only goes down. Thus, after at most $|U|\cdot |W|$ steps Phase 1 is over.
 
For Phase 2 we maintain the invariant that no matched $w \in W$ is part of a social blocking pair in any step of the phase. Assume conversely that at some point in Phase 2 there is some matching edge $\{u,w\}$ where $w \in W$ is part of a social blocking pair $\{u',w\}$. As Phase 1 ends only when no matched $w$ can improve further, this situation has to occur after some social blocking pair $\{w',u''\}$ has been resolved in Phase 2. But as $w$ is still matched, this matching edge does not influence $w$'s utility. Also, $u''$ did improve and no vertex in $U$ drops in terms of utility as $w'$ was unmatched before and thus did not leave an agent when matching to $u''$. Hence, all vertices in $U$ which did not want to match to $w$ before still do not want to match $w$. Therefore no matched $w$ can be involved in a social blocking pair during Phase 2. As no matched $w\in W$ ever rematches, no $u\in U$ becomes unmatched and decreases in utility during Phase 2. Thus, in Phase 2 there can be at most $|U| \cdot |W|$ steps. The output is a socially stable matching, as there is no social blocking pair for matched (invariant) and unmatched (Phase 2 terminates) $w \in W$. 

\subsection{Proof of Theorem~\ref{thm:considerOneSided}}
Observe that if an edge $\{u,w\} \in L$ forms in $M$, then there are no further considerate blocking pairs for $u$ and $w$ throughout. Hence, if this happens, $\{u,w\}$ remains fixed throughout the run of the algorithm. This does not harm any of the subsequent arguments.

In Phase 1, we again observe that the number of matched $w\in W$ can only decrease. Also, no $w\in W$ ever rematches with some $u\in U$ to which it had been matched before as each matched $w\in W$ only switches partner if it can improve utility by doing so. Once an agent of $W$ loses its partner (due to some other vertex of $W$ matching to it), it will not be considered in Phase 1 anymore. Hence, overall Phase 1 terminates after at most $|U|\cdot |W|$ steps.

For Phase 2 we again maintain the invariant that no matched $w\in W$ is involved in a considerate blocking pair. This claim holds directly after Phase 1 ended. We show that if this holds before some considerate blocking pair $\{u,w\}$ is resolved, then in the resulting matching it holds again. Assume conversely that after $\{u,w\}$ is resolved some matched vertex becomes part of a considerate blocking pair. As $w$ was single, he does not leave any partner in $U$ when matching with $u$. So $w$'s choice was not constrained by the links, and hence $\{u,w\}$ was an ordinary blocking pair. By picking the most preferred one, $w$ is not part of any blocking pair afterwards. Now $u$ matching with $w$ of course opens up the possibility for his former partner $w'$ (if any) to move to $u'$ with $\{u,u'\} \in L$, but this former partner is now unmatched. As there are no links between vertices in $W$, inserting matching edge $\{u,w\}$ alters only the accessible partners for $w$ and $w'$. Since $u$ increases in utility, there are also no additional (considerate) blocking pairs involving $u$. Thus, every considerate blocking pair that evolves must have been present before. This proves that Phase 2 also terminates after at most $|U|\cdot |W|$ steps. The output is a considerate matching, as there is no considerate blocking pair for matched (invariant) and unmatched (Phase 2 terminates) $w \in W$. 
\qed

\subsection{Proof of Theorem~\ref{perceivedOnesided}}
Note that in the case of matching with friendship the term \emph{most preferred blocking pair} refers to a perceived blocking pair whose resolution provides the largest perceived welfare.

In phase 1 the number of matched $w\in W$ can only decrease and no $w\in W$ ever rematches with some $u\in U$ it has been matched to before. Once an agent of $W$ loses its partner (due to some other vertex of $W$ matching to it), it will not be considered in Phase 1 anymore. Also, each matched $w\in W$ only switches partner if it can improve perceived utility by doing so. Due to the structure of $\alpha$, the perceived benefit for $w$ does only result from its direct matching partner. Thus $w$ only switches $u$ to $u'$ if $b_w(u,w)<b_w(u',w)$. Hence, $w$ can only be involved in at most $|U|$ resolutions of perceived blocking pairs. Overall, Phase 1 again terminates after at most $|U|\cdot |W|$ steps.

Phase 2 becomes slightly more complicated to analyze. We maintain the invariant that no matched $w\in W$ is involved in a perceived blocking pair. This claim holds directly after Phase 1 ended. We show that if this holds before some perceived blocking pair $\{u,w\}$ is resolved, then in the resulting matching it holds again. Assume conversely that after $\{u,w\}$ is resolved some matched vertex becomes part of a perceived blocking pair. We know that $w'\neq w$ as $u$ was $w$'s most preferred blocking pair partner. Thus, $w'$ was already matched to some $u'\neq u$ before $\{u,w\}$ was resolved but only becomes involved in a perceived blocking pair $\{u'',w'\}$ now. Let $M$ be the matching before $\{u,w\}$ is resolved and $M'$ the matching resulting resolving $\{u,w\}$. Also, let $M''$ be the matching resulting from $M'$ when resolving $\{w',u''\}$, and $\tilde{M}''$ the matching resulting from $M$ if we add $\{u,w\}$ and delete all adjacent edges (that is, resolve $\{u,w\}$ although it might not be a blocking pair). As for $w'$ we have $B_p(M,w')=B_p(M',w')$ and $B_p(M'',w')=B_p(\tilde{M}'',w')$, $w$ is already willing to switch in $M$. Thus $u''$ must not want to switch in $M$ but in $M'$. First assume that $u''=u$. Then $B_p(M,u)<B_p(M',u)$ and further $M''=\tilde{M}''$. Hence, if $u''$ is willing to switch in $M'$, the same holds for $M$. Now assume that $u''\neq u$. Then $u''$ might receive perceived benefit from $u$ which changes from $M$ to $M'$. Note that, as $w$ was unmatched before (that is, $w$ did not leave some benefit-providing $\tilde{u}$ for $u$) and $u''$ does not receive perceived benefit from any vertex in $W$, this is the only perceived benefit that changes for $u''$ between $M$ and $M'$. But then $B_p(M')-B_p(M'')=b_{u''}(\{u'',w'\})-\alpha_{u'',u'}b_{u'}(\{u',w'\})=B_p(M)-B_p(\tilde{M}'')$. Thus, again $u'$ has the same incentive to switch in $M$ as in $M'$. Next, we realize that if some $u\in U$ is matched to some $w$, it is only willing to switch to some unmatched $w'$ if $b_{u}(\{u,w'\})>b_{u}(\{u,w\})$. Thus, once a vertex $u\in U$ is matched in Phase 2, it never becomes unmatched again (as no matched $w\in W$ wants to switch). In every rematching step $u$ increases its direct benefit, so $u$ can only be involved in at most $|W|$ resolutions of perceived blocking pairs. This proves that Phase 2 also terminates after at most $|U|\cdot |W|$ steps. The output is a friendship matching, as there is no perceived blocking pair for matched (invariant) and unmatched (Phase 2 terminates) $w \in W$. 
\qed

\end{document}